\documentclass[11pt]{article}

\usepackage[top=2cm, bottom=4.4cm, left=2.5cm, right=2.5cm]{geometry}

 \usepackage[utf8]{inputenc}
\usepackage[dvipsnames]{xcolor}
\usepackage{graphicx}
\usepackage{hyperref}
\usepackage{graphicx} 
\usepackage{amsmath}
\usepackage{amssymb}
\usepackage{amsthm}
\usepackage{physics}
\usepackage{enumitem}
\usepackage{cleveref}
\usepackage{booktabs}
\usepackage{hyperref}
\usepackage{stmaryrd}
\usepackage{mathtools}
\usepackage{algorithm}
\usepackage{algorithmic}
\usepackage[normalem]{ulem}
\usepackage[export]{adjustbox}

\newtheorem{theorem}{Theorem}[section]
\newtheorem{definition}[theorem]{Definition}
\newtheorem{lemma}[theorem]{Lemma}
\newtheorem{corollary}[theorem]{Corollary}
\newtheorem{remark}[theorem]{Remark}

\newtheorem{informal}{Informal definition}

\usepackage[colorinlistoftodos]{todonotes}

\newcommand{\HH}{\mathcal{H}}
\newcommand{\Id}{\mathrm{Id}}

\newcommand{\CC}{\mathbb{C}}

\newcommand{\NN}{\mathbb{N}}

\newcommand{\BB}{\mathcal{B}}

\newcommand{\UU}{\mathcal{U}}

\newcommand{\Ss}{\mathcal{S}}
\newcommand{\Ff}{\mathcal{F}}
\newcommand{\OO}{\mathcal{O}}

\newcommand{\probP}{\text{I\kern-0.15em P}}

\newcommand{\RE}{\mathcal{R}}

\newcommand{\DD}{\mathcal{D}}

\newcommand{\CI}{\mathcal{C}}

\newcommand{\bea}{\begin{eqnarray}} 
\newcommand{\eea}{\end{eqnarray}}

\newcommand{\negl}{\mathsf{negl}}
\newcommand{\poly}{\mathsf{poly}}
\newcommand{\polylog}{\mathsf{polylog}}
\usepackage{authblk} 

\begin{document}

\title{Quantum pseudoresources imply cryptography}

\author{Alex B. Grilo}
\author{Álvaro Yángüez \thanks{alvaro.yanguez@lip6.fr}}
\affil{\textit{Sorbonne Université, CNRS, LIP6, 4 Place Jussieu, Paris F-75005, France}}

\date{}

\maketitle
\begin{abstract}
\setlength{\parindent}{15pt}

While one-way functions (OWFs) serve as the minimal assumption for computational cryptography in the classical setting, in quantum cryptography, we have even weaker cryptographic assumptions such as pseudo-random states, and EFI pairs, among others. Moreover, the minimal assumption for computational quantum cryptography remains an open question. Recently, it has been shown that pseudoentanglement is necessary for the existence of quantum cryptography (Goulão and Elkouss 2024), but no cryptographic construction has been built from it.

\indent In this work, we study the cryptographic usefulness of quantum pseudoresources —a pair of families of quantum states that exhibit a gap in their resource content yet remain computationally indistinguishable. We show that quantum pseudoresources imply a variant of EFI pairs, which we call EPFI pairs, and that these are equivalent to quantum commitments and thus EFI pairs. Our results suggest that, just as randomness is fundamental to classical cryptography, quantum resources may play a similarly crucial role in the quantum setting.

\indent Finally, we focus on the specific case of entanglement, analyzing different definitions of pseudoentanglement and their implications for constructing EPFI pairs. Moreover, we propose a new cryptographic functionality that is intrinsically dependent on entanglement as a resource.

\end{abstract}

\clearpage

\section{Introduction}

Many quantum resources are fundamental to achieving quantum advantage in information-processing tasks over conventional classical devices, e.g., entanglement \cite{PV06,HHHH09}, coherence \cite{BCP14,SAP17}, or "magic" \cite{VHGE14,MC17}. However, the manipulation of physical systems to operate with these resources is constrained in terms of time and space. Ignoring these computational limitations leads to an impractical characterization of such resources.

Inspired by complexity theory, a recent line of research studies these resources from a computational perspective. This phenomenon, known as \textit{pseudoresources}, characterizes states that do not possess a given resource yet "look like" resourceful states to a computationally bounded observer~\cite{HBE24,BMB+24}. Among the different resources, the concept of \textit{pseudoentanglement}~\cite{ABF+23,ABV23,GE24,LREJ25} stands out as a key example, where entanglement is the ``hidden'' resource of interest. More concretely, it describes the case where two families of states exhibit a large gap in the amount of their entanglement yet remain indistinguishable to a polynomial-time quantum adversary.

In classical cryptography, the resource of randomness plays a crucial role. Moreover, its computational variant, \textit{pseudorandomness}, is at the core of symmetric cryptographic constructions. This naturally raises the question of the role quantum resources play in quantum cryptography. How can pseudoresource states be constructed from other cryptographic primitives? Which cryptographic functionalities can be implemented given the existence of pseudoresources?

In the specific case of pseudoentanglement, various constructions have been proposed, ranging from those based on One-Way Functions (OWFs)\footnote{In short, OWFs are functions that are easy to compute but hard to invert.} for pure states \cite{ABF+23,ABV23,LREJ25} to EFI pairs\footnote{An EFI pair consists of two efficiently generated quantum states which are far in trace distance, but which are indistinguishable by computationally bounded adversaries. See \Cref{def:EFI} for a formal definition} for mixed states \cite{GE24}. In particular, this latter construction establishes pseudoentanglement as a minimal assumption for the existence of most computationally based quantum cryptographic protocols. However, to the best of our knowledge, no cryptographic primitive has yet been constructed directly from pseudoentangled states.

Beyond entanglement, other quantum resources have also been explored from a computational complexity perspective. For instance, recent work has examined magic states \cite{GLG+24} and coherence \cite{HBE24} in the context of complexity theory. Moreover, based on \textit{pseudorandom density matrices} (PRDMs), these resources have also been studied in the case of mixed states \cite{BMB+24}. PRDMs, which represent density matrices that are computationally indistinguishable from Haar random ones, and pseudomagic pure states have both been shown to imply EFI pairs \cite{BMB+24,GLG+24}.\footnote{Actually, pseudomagic implies EPFI pairs, that we define in this work.} 

In this work, we take a step further and demonstrate that quantum pseudoresources can be leveraged to construct useful cryptographic primitives. To do so, we introduce an extension of EFI pairs, which we call \textit{EPFI pairs}, and show that they imply quantum commitment in a way similar to how EFI pairs do. More importantly, we present a general method for constructing EPFI pairs from quantum pseudoresources. As a corollary, this establishes that quantum pseudoresources can be used to construct several cryptographic primitives, including commitments, oblivious transfer, and secure multiparty computation.

\subsection{Our contribution}

We describe now our contributions in more details.

\subsubsection{Definition of EPFI pairs and constructing commitment schemes.}

The cryptographic primitive of EFI pairs (Efficiently generated, statistically far and indistinguishable states) was first defined in \cite{BCQ23}, and it consists of a pair of states $\rho$ and $\sigma$ that are far in trace distance but cannot be efficiently distinguished by polynomial-time algorithms.  In \cite{BCQ23}, they showed that EFI pairs are equivalent to quantum commitments, more specifically, to the statistically binding variant of canonical quantum commitments introduced in \cite{Yan22}. Consequently, EFI pairs serve as a fundamental assumption for the existence of commitments, oblivious transfer, multiparty computation, and computational zero-knowledge proofs for non-trivial languages.

In order to achieve our results, we need to slightly modify such a definition as follows. We consider two keyed families of states $\{\rho_{k}\}_k$ and $\{\sigma_{k'}\}_{k'}$, and we require that for every $k$ and $k'$, $\rho_{k}$ is far from $\sigma_{k'}$ in trace distance, but are still indistinguishable by efficient algorithms. We call such pair of ensembles as EPFI pairs for Efficiently generated, pairwise statistically far and computationally indistinguishable families. We notice that an EFI pair can be seen as an EPFI pair in which each ensemble contains one element. However, EPFI pairs do not trivially imply EFI pairs: we can have families of states that are pairwise far but whose mixture is close.

In \Cref{sect:EPFI-QC}, we demonstrate that the existence of EPFI pairs of ensembles implies the existence of quantum commitments. Informally, a commitment scheme is a two-party cryptographic primitive in which a committer commits to a bit that remains hidden from the receiver until the committer chooses to reveal it. The scheme must satisfy two properties: \textit{binding}, meaning the committer cannot change the committed value, and \textit{hiding}, meaning the receiver cannot learn the value before the reveal phase.

We focus on the canonical quantum commitments introduced in \cite{Yan22}, where it is shown that proving either binding or hiding in the \textit{semi-honest} setting—where both parties follow the protocol during the commitment phase—is sufficient. This model simplifies analysis by restricting adversarial behavior to the reveal phase and was proven in \cite{Yan22} to be equivalent to stronger notions of binding, such as sum-binding \cite{Unr16}. For completeness, we briefly explain this construction.

As with EFI pairs, EPFI pairs naturally lead to the construction of canonical quantum commitments. 
\begin{itemize}
\item \textbf{Commit stage:} The committer commits to a bit $b$ by generating a bipartite state $\ket{\psi^b}_{CR}$ and sending the register $C$ to the receiver.
\item \textbf{Reveal stage:} The committer discloses register $R$ along with the bit $b$, allowing the receiver to verify the commitment by projecting onto $\ket{\psi^b}_{CR}$. If the verification fails, the receiver aborts.
\end{itemize}

In our construction, we introduce a classical key to accommodate the use of state ensembles. This modification does not alter the protocol since the key is also revealed in the opening stage. Our approach closely follows the honest statistical binding and computational hiding canonical quantum commitments construction from EFI pairs of \cite{BCQ23}. Here, the committed register $C$ contains a state sampled from one of the EPFI pairs. In the reveal stage, the committer provides the receiver with the purification of the sampled state, the secret key associated with the state, and the committed bit.

Honest binding follows from the statistical distance between any state sampled from one family of the EPFI pair and all states from the opposing ensemble together with Uhlmann's theorem \cite{Uhl76}. Similarly, computational hiding follows from the computational indistinguishability of EPFI pairs. Thus, just as EFI pairs yield a statistically binding canonical quantum commitment when restricted to two algorithms, EPFI pairs provide the same commitment structure when generalized to two families of algorithms for two possible committed values. This new functionality, EPFI pairs, it is going to play a central role in the construction of cryptography from different quantum pseudoresources, as sketched in \Cref{fig:primitives}.

\subsubsection{EPFI from quantum pseudoresources.} 
We explore the role of pseudoresourced states in quantum cryptography and demonstrate how they can lead to the construction of EPFI pairs and thus, cryptography, as represented in \Cref{fig:primitives}. In quantum information theory, a resource refers to any intrinsic property of a quantum system—such as entanglement, magic, or coherence—that provides an advantage for information processing tasks. Each resource is characterized by a corresponding set of free (resourceless) states and the free operations that cannot generate the resource. Various measures quantify these resources by evaluating, informally, how “far” a state is from the set of free states. One particularly relevant measure, and the one we adopt in this work, is the \textit{relative entropy of resource}.

\begin{figure}[t!]
    \centering
    \includegraphics[width=1.1\textwidth, left]{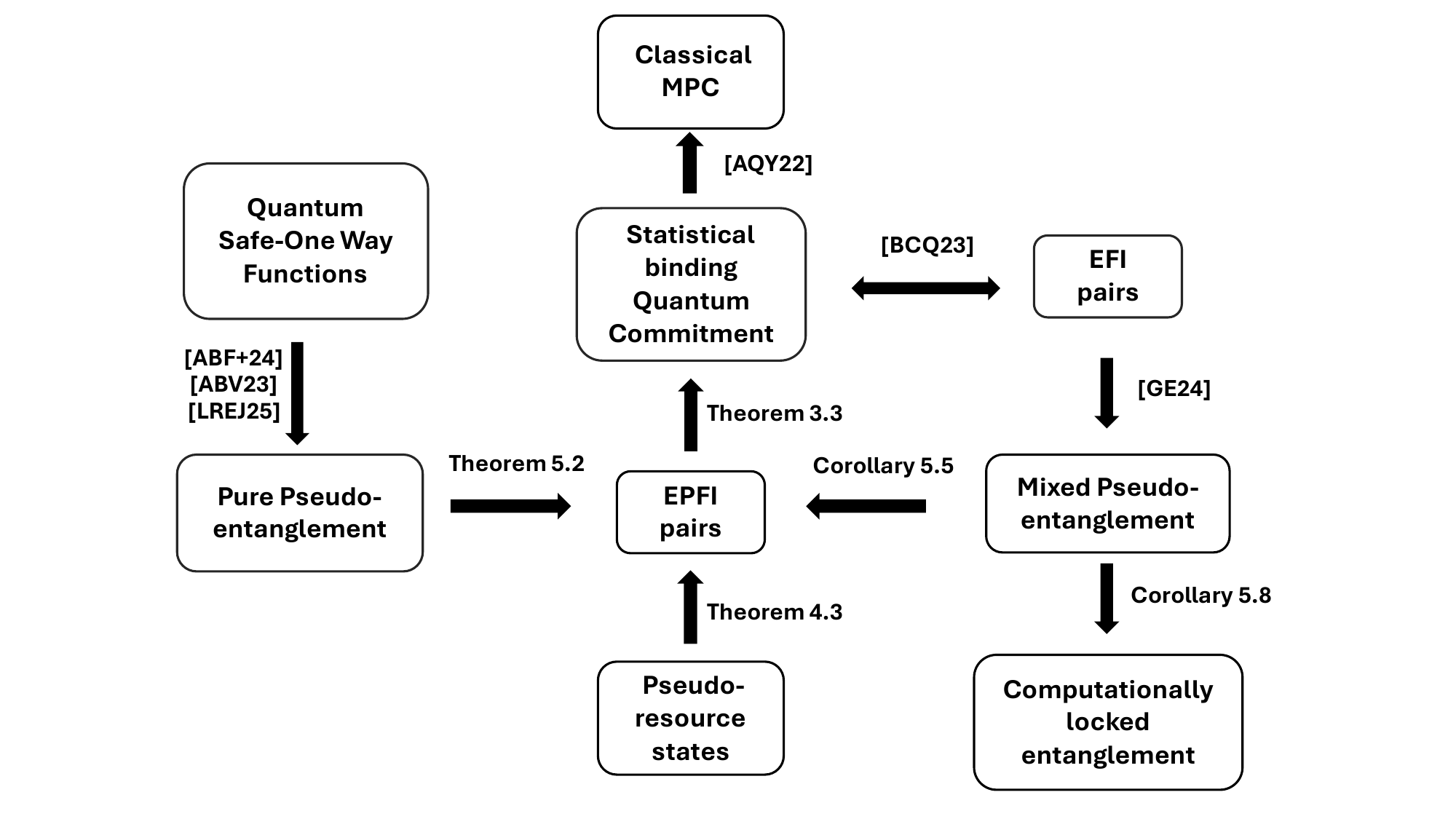}
    \caption{Summary of results and its relation with previous defined primitives.}
    \label{fig:primitives}
\end{figure}

\begin{informal}[$\eta$-gapped pseudoresource]
   A pair of efficiently generated families of quantum states is said to have an $\eta$-gapped pseudoresource if there is at least an $\eta$-gap in the relative entropy of the resource of the states sampled from each one of the families, but these families are computationally indistinguishable.
\end{informal}

Therefore, under computational restrictions, the inherent resource properties of quantum states can be effectively concealed. This phenomenon has been extensively studied in the contexts of entanglement \cite{ABF+23,ABV23,GE24,LREJ25}, magic \cite{GLG+24}, and coherence \cite{HBE24}. More generally, works such as \cite{HBE24,BMB+24} have characterized pseudoresources for any resource monotone. In our approach, we focus on the relative entropy of resource—though the definition can be extended to any asymptotically continuous monotone function —to formalize the notion of a pseudoresource.

We construct EPFI pairs by assuming the existence of a pair of families of states with an $\eta$-gapped pseudoresource. The construction follows naturally: each family in the pseudoresourced pair directly corresponds to a family in the EPFI pair. The efficient generation and computational indistinguishability properties of EPFI pairs are immediate from the definition of an $\eta$-gapped pseudoresource. However, establishing statistical distance requires further analysis. Our approach relies on an inequality from \cite{Win16} that relates the relative entropy of resources to trace distance. Informally, statistical farness follows from the asymptotic continuity of the resource measure: when two states exhibit a significant gap in their resource value (in this case, relative entropy of resource), they must also exhibit a large gap in trace distance.

\subsubsection{Pseudoentanglement} 
Entanglement is the central resource in quantum information theory. In the study of resourcefulness under computational constraints, \textit{pseudoentanglement} has received the most attention, leading to various constructions and definitions.

The first formalization of pseudoentanglement, introduced in \cite{ABF+23}, is restricted to families of pure states. This limitation arises not only from the construction itself but also from the choice of entanglement measure, as the entropy of entanglement lacks a well-defined operational meaning for mixed states. While later works, such as \cite{ABV23,LREJ25}, consider more general entanglement measures, their constructions are still pseudoentangled families of pure states.

\begin{informal}[Pure $\eta$-gap pseudoentanglement]
A pair of efficiently generated families of quantum pure states is said to have $\eta$-gap pseudoentanglement if there is at least an $\eta$-gap in the entanglement entropy of the states sampled from each one of the families, while the ensembles are computationally indistinguishable. 
\end{informal}

The proposed definition aligns with that of \cite{ABF+23}, with a particular emphasis on the entanglement entropy gap between states sampled from each family. The construction of EPFI pairs from families exhibiting $\eta$-gapped pseudoentanglement follows a similar approach to the one used for pseudoresources. Here, the relevant states are the reduced density matrices obtained by tracing out the partition used to measure entanglement. Applying Fannes' inequality \cite{Fa73}, we establish the statistical distance of the EPFI pair, while computational indistinguishability and efficient generation properties follow directly from the definition of pseudoentanglement.

 Unlike pure pseudoentanglement, the entanglement in pseudoentangled mixed states cannot be quantified using the entanglement entropy. The necessity of an operational meaningful entanglement monotone motivated \cite{ABV23} the use of the two most used measures of entanglement: the entanglement cost, $E_C$, and the distillable entanglement, $E_D$. Moreover, they proposed the use of computationally meaningful counterparts of these measures, in which the generation or distillation of entanglement has to be implementable by a poly-time algorithm. The definition of pseudoentanglement with these entanglement monotones allowed \cite{GE24} to construct a pseudoentangled mixed state family from EFI pairs or, what it is the same, to prove mixed pseudoentanglement as a new minimal assumption for the existence of computational based cryptography, as sketched in \Cref{fig:primitives}. Moreover, the use of other entanglement monotones also allows the existence of a maximal gap of a gap in the pseudoentangled families of $\Theta(n)$ vs $0$ \footnote{In contrast to pure pseudoentanglement, where the maximal gap is given by $\Theta(n)$ vs $\omega(\log(n))$.}.

\begin{informal}[Mixed $\eta$-gap pseudoentanglement]
A pair of efficiently generated families of quantum mixed states is said to have $\eta$-gap pseudoentanglement if there is at least an $\eta$-gap in the relative entropy of entanglement of the states sampled from each one of the families, while the ensembles are computationally indistinguishable.
\end{informal}

While previous definitions of pseudoentanglement have been extended to mixed states \cite{ABV23,GE24}, our formulation introduces key distinctions. We explicitly require both families to be efficiently preparable and quantify entanglement using an information-theoretic measure. The construction of EPFI pairs from mixed $\eta$-gapped pseudoentanglement follows a similar approach to prior cases, with the main difference being the inequality used to relate the regularised relative entanglement entropy to trace distance, which in this case is derived from \cite{Win16}. Consequently, assuming a mixed $\eta$-gapped pseudoentanglement with $\eta \geq 2 +1/\poly(n)$, EPFI pairs exist.

At first glance, the use of an information-theoretic entanglement measure may seem restrictive. However, it actually represents a relaxation of previous definitions. Some families exhibit no gap in computational entanglement measures, yet EPFI pairs can still be constructed from them. Conversely, any pseudoentangled states with an $\eta \geq 2 +1/\poly(n)$ gap in their (asymptotic) computational entanglement as defined in \cite{ABV23} measures imply the existence of EPFI pairs.

\subsubsection{Computationally locked entanglement}

 The research of new quantum functionalities with no classical analogue has been recently motivated by the objective of finding the minimal assumption. Nevertheless, it is not clear which applications some of these new functionalities can have. Moreover, in the case of pseudoentanglement, it was proven and the minimal assumption \cite{GE24} but no cryptographic functionality was proposed. After the construction of canonical quantum commitments and EFI pairs from pseudoentangled states, the main question that arises is the existence of a functionality that is inherently constructed from the existence of pseudoentanglement. 

Informally, the \textit{computationally locked entanglement} functionality is given by an efficiently generated family of states $\{\psi_{AB}^{k(\lambda)}\}_{k}$ which has high distillable entanglement, i.e, $\hat{E}^{\epsilon}_D (\{k, \psi_{AB}^{k}\}) \geq d$. This entanglement is efficiently distillable given the classical (or quantum) key $k$. Nevertheless, without having access to the key, the family is computationally indistinguishable from a low entangled family $\{\psi_{AB}^{k(\lambda)}\}_{k}$, i.e., $\hat{E}^{\epsilon}_C (\{ \psi_{AB}^{k(\lambda)}\}_{k}) \leq c$, where $c < d$. This construction can be seen as a dual case of the previously defined pseudoentanglement \cite{GE24,ABF+23,ABV23}: in this case, the "highly entangled" family is the efficiently generated one while the "low-entangled" family does not have to. This functionality is also relaxation of our proposed definition for pseudoentanglement, in which both families have to be efficiently generated. Possible applications of such a functionality can be authenticated quantum teleportation or certified routing quantum networks. {We leave as an open question if such a functionality can be built from weaker computational assumptions

\subsection{Open questions}

The study of quantum resources from a limited computational perspective is a promising area of research. In the case of entanglement, \cite{LREJ25} has recently proven that in the case of pure states, the manipulation of entanglement given polynomially bounded operations diverges significantly from the unbounded framework. However, for general computational resource theories there is currently no accepted measure that captures the “amount” of a resource when one is restricted to efficient (e.g. polynomial-time) operations. A promising direction is to develop a notion of quantum relative entropy that is defined relative to a class of computationally limited operations.

Another fundamental question is the relationship between EPFI and pseudoresources with complexity classes. According to \cite{Kre21}, $\mathsf{PP} \neq \mathsf{BQP}$ is necessary for the existence of PRS; however, such a condition is not known for EFI pairs—and therefore our proposed primitives might be even weaker. Understanding whether the complexity condition differs from that of PRS would provide insights into a possible separation of MiniQCrypt into two distinct worlds.

The relation between pseudomagic and pseudoentanglement has been studied for pure states \cite{GOL24}; however, it remains unexplored for mixed states. While the construction in \cite{BMB+24} simultaneously exhibits both pseudoentanglement and pseudomagic, it is unclear whether they are independent in general. Moreover, demonstrating the possibility of constructing mixed pseudoentangled states without magic would be an intriguing result, suggesting that computational indistinguishability is independent of the resource of magic for mixed states.

Lastly, our results closely depend on Fannes-type inequalities, and proving tighter versions of them for the different resources would directly improve the security of our cryptographic constructions.

\section{Preliminaries}
\subsection{Notation}
 Quantum states are represented as density matrices $\rho \in \BB_1(\HH)$. The set of states is defined as $\Ss(\HH) = \{ \rho \in \BB_1(\HH) \, | \, \rho \geq 0, \tr\rho = 1 \}$.A bipartite entangled state is defined as a state that is not separable, i.e., it cannot be written as $\rho_{AB} = \sum_i p(i) \rho_A^i \otimes \rho_B^i.$ A maximally bipartite entangled state is given by $\ket{\Psi} = \sum_{i=0}^{d-1}\ket{ii}/\sqrt{d} \in \HH_A \otimes \HH_B$. In the case of the space of dimension $2$, a maximally entangled state is known as a Bell state, and it is of the form $\ket{\Phi}= (\ket{00} + \ket{11} )/\sqrt{2}\in \HH_A \otimes \HH_B$, with $\HH_A = \HH_B = (\CC^{2})^{\otimes n}$. We denote by $\UU(\HH)$ the set of unitary operators.
The fidelity of two states $\rho$ and $\sigma$ is given by $F(\rho,\sigma) =[ \Tr ( \sqrt{\sqrt{\rho}\sigma\sqrt{\rho}})]^2$.

 The probability of distinguishing two density matrices is upper-bounded by $\frac{1}{2}(1+\Delta(\rho,\sigma))$, where $\Delta(\rho,\sigma)$ is the trace distance. The states $\rho$ and $\sigma$ are said to be statistically close when a negligible function $\mu(\lambda)$ exists such that $\Delta(\rho,\sigma) \leq \mu(\lambda)$. Given a \textit{security parameter} $\lambda$, $\mu(\lambda)$ is $\negl(\lambda)$,i.e., negligible, if, for every fixed $c$, $\mu(\lambda) = o(1/\lambda^c)$.

In the Landau notation, given two functions $f(n)$ and $g(n)$, we write $f(n)=o(g(n))$ if \newline $\lim_{n\rightarrow \infty} f(n)/g(n) = 0$. In the same way, $f(n)=\omega(g(n))$ if $\lim_{n\rightarrow \infty} f(n)/g(n) = \infty$. $f(n)=O(g(n))$ if there exist a constant $ C>0$ such that $\lim_{n\rightarrow \infty} f(n)/g(n) \leq C$. Similarly, $f(n)=\Omega(g(n))$ if there exist a constant $ C>0$ such that $\lim_{n\rightarrow \infty} f(n)/g(n) \geq C$. Lastly, $f(n)=\Theta(g(n))$ if both $f(n)=O(g(n))$ and $f(n)=\Omega(g(n))$.

\subsection{LOCC maps}

In the study of entanglement theory, one of the main objectives is the characterization of the entanglement. In order to understand how much entanglement a quantum state has, the first step is to manipulate the quantum state. Quantum channels are completely positive and trace preserving maps that transform quantum states into quantum states. In the study of entanglement, the relevant quantum channels are the ones that does not increase (or create) entanglement. A natural class that arises from a practical perspective is the one in which the parties are locally separated, and they are only allowed to perform classical communication.

\begin{definition} [LOCC map \cite{ABV23}] A quantum channel is said to be an LOCC map 

\begin{equation*}
    \Gamma : \HH_A \otimes\HH_B \mapsto \HH_{\bar{A}} \otimes \HH_{\bar{B}}
\end{equation*}

if it can be implemented by a two-party interactive protocol where each party can implement arbitrary local quantum computations and the two parties can exchange arbitrary classical communication. 
\end{definition}

\begin{definition} [Circuit description of an LOCC map \cite{ABV23}]
Given an LOCC map $\Gamma$, its circuit description is given by two families of circuits $\{\CI_{A,i}\}_{i \in \{1,..., r\}}$ and $\{\CI_{B,i}\}_{i \in \{1,..., r\}}$, each of them acting on $n_A +  t_A + c$ and $n_B +  t_B + c$ qubits respectively, such that the following procedure implements the map $\Gamma$ on an arbitrary input $\varphi_{AB} \in (\CC^2)^{\otimes n_A} \otimes (\CC^2)^{\otimes n_B}$ : 

\begin{enumerate}
    \item Registers $A$ and $B$ of $n_A $ and $n_B$ qubits are initialized in the state $\varphi_{AB}$. Ancilla registers $A'$ and $B'$ of $t_A$ and $t_B$ qubits are initialized in the $\ket{0}$ state. Communication register $C$ is also initialized in the $\ket{0}$ state.

    \item For $i = 1,..., r$, the circuit $\CI_{A,i}$ is applied to registers $A$, $A'$ and $C$. Then, register $C$ is measured in the computational basis. Thirdly, the circuit $\CI_{B,i}$ is applied to registers $B$, $B'$ and $C$. Lastly, register $C$ is measured in the computational basis.

    \item The final output of the LOCC map in the subspace $\HH_{\bar{A}}$ correspond to the state in the registers $A$ and $A'$. In the same way, the state in the registers $B$ and $B'$ is the final output on $\HH_{\bar{B}}$.
\end{enumerate}

A family of LOCC maps $\{\hat{\Gamma}_{\lambda}\}_{\lambda \in \NN}$ is said to be efficient if there exists a polynomial $c$ such that for all $\lambda$, $\hat{\Gamma}_{\lambda}$ has a circuit description whose total number of gates, including the ancilla creation and qubit measurements, is at most $c(\lambda)$.
    
\end{definition}

\subsection{Entanglement measures}

We will now introduce some functions that allow the quantification of entanglement. The most fundamental measure of entanglement in the case of pure states is given by the entanglement entropy.

\begin{definition} [Entanglement entropy]
Given a bipartite state $\rho_{AB} = \ket{\psi}\bra{\psi}_{AB} \in \HH_A \otimes \HH_B$, its \textit{entanglement entropy} is defined as 
\begin{equation*}
    E_{A/B}(\rho_{AB}) := S(\rho_A) = S(\rho_B) \,,
\end{equation*}
where $\rho_A = \Tr_{B}(\rho_{AB})$, $\rho_B = \Tr_{A}(\rho_{AB})$ and $S(\rho) = -\Tr(\rho \log \rho) $ is the von Neumann entropy. 
\label{def:entent}
\end{definition}

Nevertheless, in the case of mixed states, the entanglement entropy is not operationally meaningful. This problem leads to several measures of entanglement \cite{PV06}. Let us introduce two of the most relevant ones.

\begin{definition}[One-shot entanglement cost \cite{ABV23}]
Let $\epsilon \in [0,1]$, $\Gamma$ be an LOCC map. The one-shot entanglement  cost of a bipartite state $\rho_{AB} \in \HH_A \otimes \HH_B$ is given by
\begin{equation*}
    E^{\epsilon}_C(\rho_{AB}) = \inf_{n, \Gamma} \left\{ n | 1 - \text{F} ( \rho_{AB} , \Gamma(\Phi^{\otimes n})) \leq \epsilon \right \},
\end{equation*}
where $F(\rho,\sigma)$ is the fidelity, $\Phi =\ket{\Phi}\bra{\Phi}$ and $\ket{\Phi} \in \HH_A \otimes \HH_B $ is a Bell pair.
\end{definition}

\begin{definition} [One-shot distillable entanglement \cite{ABV23}]
Let $\epsilon \in [0,1]$, $\Gamma$ be an LOCC map. The one-shot distillable entanglement of a bipartite state $\rho_{AB} \in \HH_A \otimes \HH_B$ is given by
\begin{equation*}
    E^{\epsilon}_D(\rho_{AB}) = \sup_{m, \Gamma} \left\{ m | 1 - \text{F} (\Gamma (\rho_{AB}) , \Phi^{\otimes m}) \leq \epsilon \right \},
\end{equation*}
where $F(\rho,\sigma)$ is the fidelity, $\Phi =\ket{\Phi}\bra{\Phi}$ and $\ket{\Phi} \in \HH_A \otimes \HH_B $ is a Bell pair.

\end{definition} 

Let us now define the asymptotic versions of the measures of entanglement, which describe the rate at which the entanglement can be extracted (or diluted).

\begin{definition}[Asymptotic IID distillable entanglement]
 The asymptotic IID distillable entanglement of a bipartite state $\rho \in \HH_A \otimes \HH_B$ is given by 

 \begin{equation*}
     E^{\infty}_D (\rho) = \inf_{\epsilon \in (0,1]} \lim_{t \rightarrow \infty} \inf \frac{1}{t} E^{\epsilon}_D(\rho_{AB}^{ \otimes t}) \,.
 \end{equation*}
 \end{definition}
 
\begin{definition}[Asymptotic IID entanglement cost]
 The asymptotic IID entanglement cost of a bipartite state $\rho \in \HH_A \otimes \HH_B$ is given by 

 \begin{equation*}
     E^{\infty}_C (\rho) = \inf_{\epsilon \in (0,1]} \lim_{t \rightarrow \infty} \sup \frac{1}{t} E^{\epsilon}_C (\rho_{AB}^{ \otimes t}) \,.
 \end{equation*}
 \end{definition}

Previous measures of entanglement are specially relevant since any entanglement measure $E(\rho_{AB})$ has to fulfill that $E^{\infty}_D (\rho_{AB}) \leq E (\rho_{AB})  \leq E^{\infty}_C (\rho_{AB})$, i.e., they are extremal \cite{DHR02}. One of the entanglement entropies that will be used in this paper is the regularised relative entropy of entanglement.

\begin{definition}[Regularised relative entropy of entanglement]
Given a bipartite state $\rho_{AB} \in \HH_A \otimes \HH_B $, its regularised entropy of entanglement is given by,
\begin{align*}
    E^{\infty}_R(\rho_{AB}) = \lim_{n \rightarrow \infty}\frac{1}{n}E_R(\rho_{AB}^{\otimes n}),
\end{align*}
where $E_R(\rho) := \min_{\sigma_{AB \in \Ss_{A:B} }} D(\rho_{AB} || \sigma_{AB}) $ is the relative entropy of entanglement and $D(\rho || \sigma) = \Tr [\rho (\log \rho -\log \sigma)]$ is the relative entropy.
\end{definition}

\subsection{Computational entanglement measures}

We have previously defined the information theoretic measures of entanglement. Nevertheless, we can restrict to the case in which the operations have to be efficiently implementable, as proposed by \cite{ABV23}. The following definitions are meaningfully defined over families of states and in the asymptotic limit.

\begin{definition}[Computational one-shot entanglement cost \cite{ABV23}]
Let $\epsilon : \NN_+ \rightarrow [0,1]$ and $\lambda \in \NN_+$. Fix polynomial functions $n_A, n_B: \NN_+ \rightarrow \NN_+$. Let $\{\rho_{AB}^{\lambda}\}_{\lambda} $ be a family of quantum states such that, for any $\lambda \geq 1$, $\rho_{AB}^{\lambda} \in \HH_A \otimes \HH_B$ is a bipartite state on $n_A(\lambda) + n_B(\lambda)$. The function $c : \NN \rightarrow \NN$ is a upper bound on the computational entanglement cost of the family $\{\rho_{AB}^{\lambda}\}_{\lambda} $, i.e. $\hat{E}_C^{\epsilon} (\{\rho_{AB}^{\lambda}\}_{\lambda}) \leq c$, if there exists an efficient LOCC map family $\{\hat{\Gamma}^{\lambda}\}_{\lambda}$ such that, for each $\lambda \geq 1$, $\hat{\Gamma}^{\lambda}$ takes an input $c(\lambda)$ EPR pairs, and
\begin{equation*}
     1 - \text{F} ( \rho^{\lambda}_{AB} , \hat{\Gamma}^{\lambda}(\Phi^{\otimes c })) \leq \epsilon ({\lambda})  , \; \; \forall k \in \NN_+ \, .
\end{equation*} 
\end{definition}

\begin{definition} [Computational one-shot distillable entanglement \cite{ABV23}] \label{def:entanglement-distillable-wo-key}
Let $\epsilon : \NN_+ \rightarrow [0,1]$ and $\lambda \in \NN_+$. Fix polynomial functions $n_A, n_B: \NN_+ \rightarrow \NN_+$. Let $\{\rho_{AB}^{\lambda}\}_{\lambda} $ be a family of quantum states such that, for any $\lambda \geq 1$, $\rho_{AB}^{\lambda} \in \HH_A \otimes \HH_B$ is a bipartite state on $n_A(\lambda) + n_B(\lambda)$. The function $d : \NN \rightarrow \NN$ is a lower bound on the computational distillable entanglement of the family $\{\rho_{AB}^{\lambda}\}_{\lambda} $, i.e. $\hat{E}_D^{\epsilon} (\{\rho_{AB}^{\lambda}\}_{\lambda}) \geq d$, if there exists an efficient LOCC map family $\{\hat{\Gamma}^{\lambda}\}_{\lambda}$ such that, for each $\lambda \geq 1$, $\hat{\Gamma}^{\lambda}$ outputs a $2d(\lambda)-$qubit state, and 
\begin{equation*}
   1 - \text{F} (  \hat{\Gamma}^{\lambda}(\rho^{\lambda}_{AB}) ,\Phi^{\otimes d}) \leq \epsilon ({\lambda})  , \; \; \forall k \in \NN_+ \,.
\end{equation*}
   
\end{definition}

We can adapt this definition to the case where, for each parameter $\lambda \in \NN$, the different possible states of size $\lambda$ are indexed by a classical key $k \in \{0,1\}^{\kappa(\lambda)}$, with $\kappa: \NN_+ \rightarrow \NN_+$. Therefore, from now on, we will refer to the family of states as $\{\rho_{AB}^k\}_{k}$. Please, note that the security parameter $\lambda$ is implicit in the size of the keys $k$. In the same way, the action of the LOCC map for the distillation of entanglement has to have access to the key $k$. Therefore, the distillation map is now given by $\Gamma( k, \rho^k)$ which input is the state $\ketbra{k}{k}_{A'} \otimes \rho^{k}_{AB}\otimes\ketbra{k}{k}_{B'}$, where the bipartition is $A'A:B'B$. This map has to efficiently distill the entanglement from states associated with all possible keys. The same applies to the cost of generating all the states associated with all possible keys. The corresponding keys of the states can be extended to the case of quantum keys (more specifically, EFI pairs), as proven by \cite{GE24}.

\begin{definition}[Uniform computational one-shot distillable entanglement \cite{ABV23}]
Let $\epsilon : \NN_+ \rightarrow [0,1]$ and $\lambda \in \NN_+$. Fix polynomial functions $n_A, n_B: \NN_+ \rightarrow \NN_+$. Let $\{\rho_{AB}^{k}\}_{k \in \{0,1\}^{\kappa(\lambda)}}$ be a family of quantum states such that, for any $\lambda \geq 1$, $\rho_{AB}^{k} \in \HH_A \otimes \HH_B$ is a bipartite state on $n_A(\lambda) + n_B(\lambda)$. The function $d : \NN_+ \rightarrow \NN_+$ is a lower bound on the computational distillable entanglement of the family $\{k, \rho_{AB}^{k}\} $, i.e. $\hat{E}_D^{\epsilon} (\{k ,\rho_{AB}^{k}\}) \geq d$ , if there exists an efficient LOCC map family $\{\hat{\Gamma}^{\lambda}\}_{\lambda}$ such that, for each $\lambda \geq 1$, $\hat{\Gamma}^{\lambda}$ outputs a $2d(\lambda)-$qubit state, and 
\begin{equation*}
     1 - \text{F} (  \hat{\Gamma}^{\lambda}(k,\rho^{k}_{AB}) ,\Phi^{\otimes d }) \leq \epsilon ({\lambda}) , \; \; \forall \lambda \in \NN_+, \; \;\forall k \in \{0,1\}^{\kappa(\lambda)}.
\end{equation*}

\label{def:cmpdistent}    
\end{definition}

\begin{definition}[Uniform computational one-shot entanglement cost\cite{ABV23}]
Let $\epsilon : \NN_+ \rightarrow [0,1]$ and $\lambda \in \NN_+$. Fix polynomial functions $n_A, n_B: \NN_+ \rightarrow \NN_+$. Let $\{\rho_{AB}^{k}\}_{k } $ be a family of quantum states such that, for any $\lambda \geq 1$ and $k \in \{0,1\}^{\kappa (\lambda)}$, $\rho_{AB}^{k} \in \HH_A \otimes \HH_B$ is a bipartite state on $n_A(\lambda) + n_B(\lambda)$. The function $c : \NN \rightarrow \NN$ is a upper bound on the computational entanglement cost of the family $\{k, \rho_{AB}^{k}\} $, i.e., $\hat{E}_C^{\epsilon} (\{k ,\rho_{AB}^{k}\}) \leq c$, if there exists an efficient LOCC map family $\{\hat{\Gamma}^{\lambda}\}_{\lambda}$ such that, for each $\lambda \geq 1$, $\hat{\Gamma}^{\lambda}$ takes an input  $c(\lambda)$ EPR pairs, and
\begin{equation*}
     1 - \text{F} (  \rho^{k}_{AB} , \hat{\Gamma}^{\lambda}(k,\Phi^{\otimes c })) \leq \epsilon ({\lambda}) , \; \; \forall \lambda \in \NN_+, \; \; \forall k \in \{0,1\}^{\kappa(\lambda)} .
\end{equation*} 
\label{def:comptentcost}    
\end{definition}

\subsection{Auxiliary lemmas}
Let us now introduce the quantum information theory tools that we make use of.

\begin{theorem} [Holevo-Helstrom \cite{Hol73,Hel69}]\label{th:HolHel}
Given two mixed states $\rho$ and $\sigma$, the best success probability to distinguish them is given by $\frac{1}{2}(1 + \Delta(\rho, \sigma))$, where $\Delta(\rho, \sigma) = \frac{1}{2} \norm{\rho-\sigma}_1$. Moreover, given $n$-copies,
\begin{equation}
    \Delta(\rho^{\otimes n}, \sigma^{\otimes n}) \geq 1- \exp(-n\Delta(\rho, \sigma)/2) \,.
\end{equation}
\end{theorem}

\begin{theorem}[Uhlmann's theorem \cite{Uhl76}]\label{the:Uhlman}
Let $\rho \in \BB_1(\HH_1)$ and $\sigma \in \BB_1(\HH_1)$ be a pair of density operators, where $\rho = \Tr_{2}(\ketbra{\psi}{\psi})$ for $\ket{\psi} \in \BB_1(\HH_1 \otimes \HH_2)$. It holds that $F(\rho, \sigma) = \max \big\{ |\langle \psi | \eta\rangle| :  \ket{\eta} \in \BB_1(\HH_1 \otimes \HH_2) \; \textit{is a pure state s.t. } \Tr_2 (\ketbra{\eta}{\eta}) = \sigma   \big\}$.
\end{theorem}

\begin{lemma} [Fannes inequality \cite{Fa73}]\label{lem:Fannes}
Given two density operators $\rho \in \BB_1(\HH)$ and $\sigma\in \BB_1(\HH)$, it holds that

\begin{equation}
    \abs{S(\rho)-S(\sigma)} \leq 2 \Delta(\rho,\sigma) \log d + c ( \Delta(\rho,\sigma)) \, ,
\end{equation}
    where $c(x) := \min \{-x\log x, 1/2e\}$ and $S(\rho) = -\Tr(\rho \log \rho) $ is the von Neumann entropy.
\end{lemma}

This well known inequality of information  theory can be extended to the quantum relative entropy,

\begin{lemma}[Fannes type inequality for the quantum relative entropy  \cite{Win16}]
For a closed, convex and bounded set $\Ff$ of positive semidefinite operators, containing at least one full rank operator, let 
\begin{align*}
    \kappa = \sup_{\tau, \tau'} D_C(\tau)- D_C(\tau')
\end{align*}
be the largest variation of $D_C(\tau):= \min_{\gamma \in C} D(\tau||\gamma)$, where $D(\rho||\gamma) = \Tr \rho(\log (\rho) - \log (\gamma))$ is the quantum relative entropy. Then, for any two states $\rho$ and $\sigma$ with $\Delta(\rho,\sigma) \leq \epsilon$,

\begin{align}
    |D_C(\rho)-D_C(\sigma)| \leq \epsilon \kappa +(1+\epsilon)h\Big(\frac{\epsilon}{1+\epsilon}\Big) \, ,
    \label{eq:FanResource}
\end{align}
where $h(\cdot)$ is the binary entropy.
\end{lemma}
Moreover, in the specific case of entanglement, it can be related to the relative entropy of entanglement,

\begin{lemma}[Fannes type inequality for the regularised relative entropy of entanglement \cite{Win16,DH99,Chr06}]
Given two states $\rho_{AB}, \sigma_{AB} \in \HH_A \otimes \HH_B$, the difference on their regularised relative entropy of entanglement is upper bounded by
\begin{align*}
    |E^{\infty}_R(\rho_{AB}) -E^{\infty}_R(\sigma_{AB})| \leq \epsilon \log d + (1 + \epsilon) h\Big( \frac{\epsilon}{1+\epsilon} \Big) \,,
\end{align*}
where $\Delta(\rho,\sigma) \leq \epsilon$   and $h(\cdot)$ is the binary entropy. 
\label{lem:mixFannes}
\end{lemma}

\subsection{Quantum Commitment}

We focus on the construction of \textit{canonical quantum commitments}, which is defined as follows.

\begin{definition}[Canonical quantum commitment \cite{Yan22}] \label{def:comm}
Given an ensemble of polynomial-time uniformly generated quantum circuit pair $\{(Q_{\lambda,0},Q_{\lambda,1})\}_{\lambda}$, a \textit{canonical quantum commitment} scheme is defined by the following stages:

\begin{itemize}
    \item \textbf{Commit stage}: the committer chooses the committed bit $b \in \{0,1\}$ and performs the quantum circuit $Q_{\lambda,b}$ to the register pair $(C,R)$, initialized in the $\ket{0}$ state\footnote{For simplicity we write the tensor product of $k$ registers in the $\ket{0}$ state as $\ket{0}$ (instead of $\ket{0}^{\otimes k}$), when it is clear from the context.}. Then the committer sends the register $C$ to the receiver, i.e., the state $\rho_{\lambda,b}:= \Tr_R{(Q_{\lambda,b}\ket{0}\bra{0}_{CR}Q^{\dagger}_{\lambda,b})}$.

    \item \textbf{Reveal stage}: the committer sends the bit $b$ and the register $R$ to the receiver. The receiver performs $Q^{\dagger}_{\lambda,b}$ on $(C,R)$ and aborts if the measured registers are not in the $\ket{0}$ state.
\end{itemize}
\end{definition}

\begin{definition}[Computational hiding]
Given a canonical quantum commitment scheme in which the committed states are given by the families of mixed states $\{\rho_{\lambda,0}\}_{\lambda}$ and $\{\rho_{\lambda,1}\}_{\lambda}$, the scheme is computationally hiding if, for any non-uniform QPT distinguisher $\DD $ with advice $\sigma_{\lambda}$ and any $m \in \poly(\lambda)$, there exists a negligible function $\nu (\lambda) > 0$ such that:
    \begin{equation*}
        \left| \probP\left[\DD(\sigma_{\lambda},\rho_{\lambda,0}^{\otimes m}) = 1] - \probP[\DD( \sigma_{\lambda}, \rho_{\lambda,1}^{\otimes m}) = 1\right]  \right| \leq \nu (\lambda) \,.
    \end{equation*}
\end{definition}

\begin{definition}[Honest statistical binding \cite{Yan22}] \label{def:stbind}
A canonical quantum commitment scheme satisfies honest statistical binding if, for any auxiliary state $\ket{\psi}$ and any unitary $U \in \UU(\HH)$ there exists a negligible function $\nu (\lambda) > 0$ such that: 
\begin{equation}
    \norm{\Big(Q_{\lambda,1} \otimes \Id_Z\ket{0}\bra{0}_{CR}Q^{\dagger}_{\lambda,1} \otimes \Id_Z \Big) \Big( \Id_C \otimes U_{RZ} \Big) \Big(Q_{\lambda,0} \ket{0}_{CR} \otimes \ket{\psi}_Z  \Big)}_2 \leq \nu(\lambda) \, .
\end{equation}    
\end{definition}

Informally, the honest-binding property allows the receiver to cheat only in the reveal stage of the commitment functionality. As proven by \cite{Yan22}, honest binding in the canonical quantum commitment is equivalent to the notion of sum-binding \cite{Unr16}. Moreover, since sum-binding is equivalent to AQY binding \cite{AQY21,MY22}, oblivious transfer and multiparty computing can be constructed from it \cite{BCKM21,GLSV21}. From now on, we will refer to the honest statistical binding commitments as statistical binding commitments.

We notice that the construction of \cite{Yan22} can be modified with commitments generated by a pair of uniformly generated families of quantum circuits  $\left(\{Q^{k(\lambda)}_0\}_{k \in \{0,1\}^{\kappa(\lambda)}}, \{Q^{k'(\lambda)}_1\}_{k' \in \{0,1\}^{\kappa(\lambda)}}\right)$.

To commit to a bit $b\in\{0,1\}$, the committer first chooses a secret key $k$ (or $k'$) uniformly at random and then applies the corresponding circuit $Q^{k(\lambda)}_{b}$, preparing the state $|\Psi^{k(\lambda)}_{b}\rangle_{CR} = Q^{k(\lambda)}_b\,\ket{0}_{CR}$. Then, they send the commitment register $C$ as in \Cref{def:comm}. In the reveal phase, the committer sends to the receiver the chosen key in the reveal stage together with $b$ and the register $R$. 

The statistical binding property in this setting requires that for all $k, k' \in \{0,1\}^{\kappa(\lambda)}$,
\begin{align}
    \norm{\Big(Q^{k'(\lambda)}_{1} \otimes \Id_Z\ket{0}\bra{0}_{CR}Q^{\dagger \,k'(\lambda)}_{1} \otimes \Id_Z \Big) \Big( \Id_C \otimes U_{RZ} \Big) \Big(Q^{k(\lambda)}_{0} \ket{0}_{CR} \otimes \ket{\psi}_Z  \Big)}_2 \leq \nu(\lambda) \, ,
    \label{eq:bindingcomm}
\end{align}

Likewise, the computational hiding property in this case is given by,
\begin{align}
          \left | \probP_{k(\lambda) }        \left[\DD(\sigma_{\lambda},\rho_{k(\lambda),0}^{\otimes m}) = 1\right] -
        \probP_{k'(\lambda) }  \left[\DD( \sigma_{\lambda}, \rho_{k'(\lambda),1}^{\otimes m}) = 1\right]  
        \right | \leq \nu (\lambda) \,.
        \label{eq:hidingcomm}
\end{align}

When using families of algorithms for committing instead of a pair of algorithms, the condition on the binding is stronger, i.e., the committing states have to be pairwise honest statistically biding. In this case, the definition of honest statistical binding for families of quantum circuits (Eq. \ref{eq:bindingcomm}) implies the definition given by \Cref{def:stbind}. Moreover,  since we have an honest committer, we still consider the mixture over the keys in the hiding property.

\section{EPFI pairs imply quantum commitments}
\label{sect:EPFI-QC}

In this section, we focus on the definition of EPFI and show how to construct 
quantum commitments from it.

We start by recalling the definition of EFI pairs~\cite{BCQ23}.

\begin{definition}[EFI pair]
    A pair of mixed states $(\rho_{0, \lambda}, \rho_{1, \lambda})$ is an EFI pair,~if
    \begin{itemize}
        \item \textbf{Efficient generation}: there exists a QPT algorithm A that on input $(1^{\lambda}, b)$ outputs $\rho_{b, \lambda}$.
        \item \textbf{Statistical distance}: $\Delta(\rho_{0, \lambda}, \rho_{1, \lambda}) \geq \Omega\left(\frac{1}{\poly (\lambda)}\right)\, .$
        \item \textbf{Computationally indistinguishability}:for any non-uniform QPT distinguisher $\DD $ with advice $\sigma_{\lambda}$ and any $m \in \poly(\lambda)$, there exists a negligible function $\nu (\lambda) > 0$ such that:
    \begin{equation*}
        \left| \probP\left[\DD(\sigma_{\lambda},\rho_{\lambda,0}^{\otimes m}) = 1] - \probP [\DD( \sigma_{\lambda}, \rho_{\lambda,1}^{\otimes m}) = 1\right]  \right| \leq \nu (\lambda)\, .
    \end{equation*}
    \end{itemize}
    \label{def:EFI}
\end{definition}

As aforementioned, it was proven that EFI pairs of states are equivalent to quantum commitments \cite{BCQ23}. 
In this work, we define EPFI pairs by refining the requirements for EFI pairs.

\begin{definition}[EPFI pair]
A pair of ensembles of mixed states $\left(\{\psi_{k(\lambda)}\}_{k(\lambda)}, \{\phi_{k'(\lambda)}\}_{k'(\lambda)}\right)$ indexed by $k, k' \in \{0,1\}^{\kappa(\lambda)}$ is an \textbf{E}fficiently generated, \textbf{p}airwise \textbf{f}ar and computational \textbf{i}ndistinguishable pair (EPFI pair), if
    \begin{itemize}
        \item \textbf{Efficient generation}: Given $k(\lambda)$ (or 
          $k'(\lambda)$),  there exists a QPT algorithm A that on input $(1^{\lambda}, k,  b)$ (or $(1^{\lambda}, k', b)$) outputs $ \psi_{k(\lambda)}$ (or $ \phi_{k'(\lambda)}$). 
        \item \textbf{Pairwise statistically far}: For all $ k, k' \in \{0,1\}^{\kappa(\lambda)}$,
        \begin{align*}
            \Delta(\psi_{k(\lambda)}, \phi_{k'(\lambda)}) \geq
          \Omega\left(\frac{1}{\poly (\lambda)}\right)\,.
        \end{align*}  
        \item \textbf{Computationally indistinguishability}: for any non-uniform QPT distinguisher $\DD $ with advice $\sigma_{\lambda}$ and any $m \in \poly(\lambda)$, there exists a negligible function $\nu (\lambda) > 0$ such that:
    \begin{equation*}
        \left | \probP_{k(\lambda) }        \left[\DD(\sigma_{\lambda},\psi_{k(\lambda)}^{\otimes m}) = 1\right] -
        \probP_{k'(\lambda) }  \left[\DD( \sigma_{\lambda}, \phi_{k'(\lambda)}^{\otimes m}) = 1\right]  
        \right | \leq \nu (\lambda) \, .
    \end{equation*}
    \end{itemize}
    \label{def:EPFInsem}
\end{definition}

The main difference between both primitives is that each element of each ensemble has to be statistically far in terms of trace distance from every element of the other ensemble, while computationally indistinguishability holds for the ensemble itself.

We can now prove the main result of this section, which is the construction of a
canonical quantum commitment from EPFI.

\begin{theorem}Assuming the existence of EPFI pairs of mixed states, there exists statistically binding and computationally hiding canonical quantum commitments.
    
\end{theorem}

\begin{proof}
  Let $\{Q^k_{\psi}\}_{k}$ and $\{Q^{k'}_{\phi}\}_{k'}$ be the two families of algorithms that generate the corresponding families of states $\{\psi_{k}\}_{k}$ and $ \{\phi_{k'}\}_{k'}$ \footnote{We omit the security parameter $\lambda$ for simplicity.}. The construction of the canonical quantum commitment is as follows. Let the algorithms $\{Q^k_{b, \lambda}\}_{k}$ with $b = \{0,1\}$ corresponds to the aforementioned $\{Q^k_{\phi}\}_{k}$ and $\{Q^k_{\psi}\}_k$, i.e., each family of states corresponds to one value of $b$. The committer generates $\lambda$ copies of the state by applying $\bigotimes_{i=1}^{\poly(\lambda)} (Q^k_{b, \lambda})_{i}$. Then, the committed state defined in \Cref{def:comm} is given by $\rho^{\otimes \poly(\lambda)}_{C, b}$, where $ \rho_{C, b} =  \Tr_R(Q^k_b \ketbra{0}{0}_{CR}Q^{k \; \dagger }_b) $  is a state from the EPFI ensemble (\Cref{def:EPFInsem}). For the opening, the committer sends the $\lambda$ registers $R$, together with the key $k$ \footnote{Please, note that the construction ca be extended to quantum keys, i.e., the key is a quantum state, in the subspace of the registers R.} and the committed bit $b$.

     Let us first prove honest statistical binding. Given the two states $\psi_{k}$ and $\phi_{k'}$, $\Delta(\psi_{k}, \phi_{k'}) \geq \Omega(1/\poly(\lambda))$ by definition of the EPFI. Lastly, by taking polynomially many copies of the states, $\Delta(\psi_{k}^{\otimes \poly(\lambda)}, \phi_{k'}^{\otimes \poly(\lambda)}) \geq 1- \negl(\lambda)$ by \Cref{th:HolHel}.
    Moreover, due to the fact that $(F(\rho,\sigma))^2 + (\Delta (\rho, \sigma))^2 \leq 1$,
    \begin{equation}
            F \Big(\psi_{k}^{\otimes \poly(\lambda)}, \phi_{k'}^{\otimes \poly(\lambda)}\Big) \leq \sqrt{1- \left(\Delta \left(\psi_{k}^{\otimes \poly(\lambda)}, \phi_{k'}^{\otimes \poly(\lambda)}\right)\right)^2} \leq \negl(\lambda) \, .
    \end{equation}
Therefore, by Uhlman's theorem (\Cref{the:Uhlman}), the scheme satisfies honest binding (Eq. \ref{eq:bindingcomm}).

Computational hiding follows from the computational indistinguishability property of the EPFI pairs: without the key, a distinguisher cannot infer the ensemble from which the state has been sampled, satisfying Eq. \ref{eq:hidingcomm}.

\end{proof}

\section{Pseudoresources and quantum cryptography}
\label{sect:PreS-EPFI}
In this section, we establish a foundational connection between computational-based quantum cryptography and resource theories. We begin by briefly introducing the formalism of resource theories, define the concept of \textit{pseudoresource}, and finally discuss how the existence of pseudoresources implies the existence of EFI pairs of ensembles and, thus, computational based cryptography.

\subsection{Resource theory and pseudoresources}

Quantum resource theories provide a systematic framework for studying properties
of quantum systems that are crucial for quantum information processing tasks
\cite{CG19}. A quantum resource theory $\RE = (\Ff,\OO)$ is characterized by a set of free states $\Ff$, and a set of
free operations $\OO$. The free states, defined as $\Ff(\HH) \subseteq \Ss(\HH)$, represent states that lack the resource of interest. A completely positive trace-preserving (CPTP) map $\Lambda : \BB(\HH_1) \rightarrow \BB(\HH_2)$ belongs to the set $\OO$ if, for every $\sigma \in \Ff$, it holds that $\Lambda(\sigma) \in \Ff$. 
For instance, in the case of entanglement, the free states consist of the set of
separable states, and we can pick the set of free operations as LOCC maps.

The set of properties that defines the resourced families of states that are necessary for our construction are:
\begin{enumerate}
    \item The set of free states, $\Ff(\HH)$, is convex and closed.
    \item  $\Ff(\HH)$ contains a full-rank state. 
\end{enumerate}

We notice that in our case, we only consider finite dimension states and
therefore property $1$ implies that  $\Ff(\HH)$ is bounded, which is also required. These properties are a relaxation of the \textit{Brandão-Plenio axioms}
\cite{BP10}, and they are satisfied by many resource theories such as magic, coherence, entanglement or athermality \cite{CG19}. 
 
There exist several measures for quantifying resources in quantum information theory, ranging from geometric to witness-based approaches. Typically, these measures assess the resource content of a state by evaluating its “distance” from the set of free states. In this work, we focus on entropic measures, and more specifically, on the \textit{relative entropy of resource}.

\begin{definition}[Relative entropy of resource] 
Given a state $\rho \in \Ss(\HH)$ and a set  of free states $\Ff ( \HH) \subseteq \Ss(\HH)$, its relative entropy of resource is defined as

\begin{equation}
    R_{rel}(\rho) := \min_{\sigma \in \Ff} D( \rho || \sigma) \, ,
\end{equation}
where $D( \rho || \sigma)  = \Tr[ \rho(\log(\rho)-\log(\sigma))]$ is the quantum relative entropy.
\label{def:RE-res}
\end{definition}

When taking into account computationally bounded operations, families of states that can be distinguished in the asymptotic regime might become indistinguishable against polynomially bounded quantum adversaries. This property becomes even more interesting when these states have very different properties. We focus on the case in which both families have a substantial difference in terms of resources.

\begin{definition}[$\eta$-gap pseudoresource]
Let $\lambda \in \NN_{+}$ and $\eta : \NN_{+} \rightarrow \NN_{+} $ be
  arbitrary. A pair of families $\{\psi_{k(\lambda)}\}_{k(\lambda)}$ and
  $\{\phi_{k'(\lambda)}\}_{k'(\lambda)}$ of (potentially mixed) states indexed by $k(\lambda), k'(\lambda) \in \{0,1\}^{\kappa(\lambda)}$  is said to have \textit{$\eta$-gap pseudoresourced} $\RE$ if,
    \begin{enumerate}
      \item \textbf{Efficient generation}:  Given $k(\lambda)$ (or $k'(\lambda)$),  there exists a QPT algorithm A that on input $(1^{\lambda}, k,  b)$ (or $(1^{\lambda}, k', b)$) outputs $ \psi_{k(\lambda)}$ (or $ \phi_{k'(\lambda)}$). 
    \item \textbf{Resource gap}: For all $k, k' \in \{0,1\}^{\kappa(\lambda)}$,
    \begin{align*}
        |R_{rel}(\psi_{k(\lambda)}) - R_{rel}(\phi_{k'(\lambda)})|  \geq \eta \, ,
    \end{align*}
    where $R_{rel}(\rho)$ is the relative entropy of resource $\RE$. 
    \item  \textbf{Computational indistinguishability}: For any non-uniform QPT distinguisher $\DD $ with advice $\sigma_{\lambda}$ and any $m \in \poly(\lambda)$, there exists a negligible function $\nu (\lambda) > 0$ such that:
    \begin{equation*}
        \Big | \probP_{k(\lambda) } [\DD(\sigma_{\lambda},\psi_{k(\lambda)}^{\otimes m}) = 1] - \probP_{k'(\lambda)}  [\DD( \sigma_{\lambda}, \phi_{k'(\lambda)}^{\otimes m}) = 1]  \Big | \leq \nu (\lambda) \, .
    \end{equation*}
\end{enumerate}
\label{def:PR}
\end{definition}

The concept of pseudoresource has been introduced before for generic measures of resources, also known as resource monotones \cite{HBE24,BMB+24}. Nevertheless, we focus on the concrete measure of relative entropy of resource, which is the most important entropic measure for quantum resource theories.

\subsection{Pseudoresources imply EPFI pairs}
We prove now the main technical contribution of our result.

\begin{theorem}[Pseudoresource implies EPFI pairs]
For any $\eta \geq 2+1/\poly(n)$, assuming the existence of $\eta$-gap pseudoresource with $\kappa := \sup_{\tau, \tau'} R_{rel}(\tau)- R_{rel}(\tau') = \polylog(d)$, then EPFI pairs exist.
\label{th:PREFI}
\end{theorem}

\begin{proof}
    The construction of an EPFI pair given a pair of pseudoresourced families of states is straightforward: each family of the  pseudoresourced states corresponds to an EFI ensemble. The efficient generation of the EPFI pairs together with the property of computational indistinguishability follow from the definition of $\eta$-gap pseudoresource. 

To prove statistical distance, it follows from \Cref{eq:FanResource} that, for every $k, k' \in \{0,1\}^{\kappa(\lambda)}$
    \begin{align*}
        \Delta(\psi_{k (\lambda)}, \phi_{k'(\lambda)}) \geq \frac{| R_{rel}(\psi_{k(\lambda)}) - R_{rel}(\phi_{k'(\lambda)})| - 2}{\kappa}.
    \end{align*}
Moreover, by taking into account that $\eta \geq 2 + 1/\poly(n)$ and $\kappa =  \polylog(d)$,
    \begin{align*}
        \Delta(\psi_{k (\lambda)}, \phi_{k'(\lambda)}) \geq \Omega \Big( \frac{1}{\poly(n)}\Big) \quad \forall k, k' \in \{0,1\}^{\kappa(\lambda)}.
    \end{align*}    
\end{proof}

\begin{remark}
    The construction of EPFI pairs from pseudorandom families of states is given by the asymptotic continuity of the relative entropy of resource. Nevertheless, a similar construction holds for any resource monotone which is asymptotically continuous.
\end{remark}

Therefore, just assuming the existence of pseudoresource families which relative entropy of resource presents a gap larger than $1/\poly(n)$, EFI pairs of ensembles and thus, quantum commitments and all the primitives that follow from it such as  oblivious transfer and multiparty computing can be constructed. Constructions such as the one of pseudoresources from pseudo-random density matrices \cite{BMB+24} hold\footnote{Please, note that in \cite{BMB+24} EFI pairs are constructed assuming the existence of pseudo-random density matrices.}. A similar construction was also proposed by \cite{GLG+24} for the specific case of pseudomagic\footnote{The construction of pseudomagic states of \cite{GLG+24} actually implies our proposed primitive of EPFI pairs, as it is proven in SM V.B of \cite{GLG+24}.}. Nevertheless, we generalize the result for any
pseudoresource.

\section{Cryptography from pseudoentanglement}

Since the proposal of the notion of pseudoentanglement \cite{ABF+23}, different definitions of the same phenomena have been proposed \cite{ABV23,GE24,LREJ25}, extending the definition to mixed states and using different measures of entanglement that allow to obtain a maximal separation in terms of entanglement between both families of states. Nevertheless, all definitions capture the same phenomena: given two families of states that are efficiently generated, both families present a gap in the entanglement, yet they are computationally indistinguishable.

Despite it was proven that the existence of (mixed states) pseudoentanglement as a minimal assumption for computational based cryptography \cite{GE24}, no cryptographic primitives have been constructed from it as far as we are concerned. The goal of this section is to establish a direct connection between the different definitions of pseudoentanglement and cryptography. The first subsection studies the construction of EPFI pairs from pure state pseudoentanglement, while the second extend this result to the case of mixed states. Lastly, we proposed a new functionality that it is intrinsically dependent to the resource of entanglement.

\subsection{Pure state pseudoentanglement implies EPFI pairs}

Let us first define the notion of pure state pseudoentanglement~\cite{ABF+23}.

\begin{definition} [Pure $\eta$-gap pseudoentanglement]
Let $\lambda \in \NN_{+}$ and $\eta : \NN_{+} \rightarrow \NN_{+} $ be arbitrary. A pair of ensembles of bipartite pure states $\{|\psi^{k(\lambda)}\rangle_{AB}\}_{k(\lambda)}$, $\{|\phi^{k'(\lambda)}\rangle_{AB}\}_{k'(\lambda)}$ indexed by $k, k' \in \{0,1\}^{\kappa(\lambda)}$ is said to have pure $\eta$-pseudoentanglement if,
\begin{enumerate}
    \item \textbf{Efficient generation}: Given $k(\lambda)$ (or $k'(\lambda)$),  there exists a QPT algorithm A that on input $(1^{\lambda}, k,  b)$ (or $(1^{\lambda}, k', b)$) outputs $|\psi^{k(\lambda)}\rangle_{AB}$ (or $|\phi^{k'(\lambda)}\rangle_{AB}$). 
    \item  \textbf{Entanglement gap}: For all $k, k' \in \{0,1\}^{\kappa(\lambda)}$,
        \begin{align*}
            |E(|\psi^{k(\lambda)} \rangle_{AB}) - E(|\phi^{k'(\lambda)} \rangle_{AB})| \geq \eta \, ,
        \end{align*}
        where $E(\rho)$ is the entanglement entropy.
     \item \textbf{Computational indistinguishability}: For any non-uniform QPT distinguisher $\DD $ with advice $\sigma_{\lambda}$ and any $m \in \poly(\lambda)$, there exists a negligible function $\nu (\lambda) > 0$ such that:
    \begin{equation*}
        \Big | \probP_{\ket{\varphi} \leftarrow \{|\psi^{k(\lambda)}\rangle_{AB}\}_{k(\lambda)} } [\DD(\sigma_{\lambda},\ket{\varphi}^{\otimes m}) = 1] - \probP_{\ket{\varphi} \leftarrow \{|\phi^{k'(\lambda)}\rangle_{AB}\}_{k'(\lambda)} }  [\DD( \sigma_{\lambda}, \ket{\varphi}^{\otimes m}) = 1]  \Big | \leq \nu (\lambda) \, .
    \end{equation*} 
\end{enumerate}
 \label{def:P-PES}   
\end{definition}

Having introduced the definition of pseudoentanglement, we can prove the main statement of this section.

\begin{theorem}[Pure  pseudoentanglement implies EPFI pairs]\label{lem:polypseudo}
For any $\eta \geq 1/2e+1/\poly(n)$ and assuming the existence of pure $\eta$-gap pseudoentanglement, then EPFI pairs exist.
\end{theorem} 
\begin{proof}
The construction of an EPFI pair given a pair of pseudoentangled pair of
  families of pure states is slightly different to the one of
  \Cref{th:PREFI}. We define the two families by with the reduced density matrices of the pseudoentangled states: 
  \begin{align*}
    \big\{\rho_{A,\psi}^{k(\lambda)}\big\}_{k(\lambda)} =
    \big\{\Tr_B\big(|\psi^{k(\lambda)}\rangle\langle \psi^{k(\lambda)}|_{AB}\big)\Big\}_{k(\lambda)},
    \text{ and  }
    \big\{\sigma_{A, \phi}^{k'(\lambda)} \big\}_{k'(\lambda)}=
    \big\{\Tr_B\big(|\phi^{k'(\lambda)}\rangle\langle\phi^{k'(\lambda)}|_{AB}\big)\Big\}_{k'(\lambda)}  \, ,
  \end{align*}
  where the entanglement is measured across the cut $(A:B)$. 
The properties of efficient generation and computational indistinguishability of
the EPFI pairs follow from the definition of $\eta$-gap pure pseudoentanglement.

We now prove that for all $k,k' \in \{0,1\}^{\kappa(\lambda)}$,  $\rho_{A, \psi}^{k(\lambda)}$ and $\sigma_{A,\phi}^{k'(\lambda)}$  are statistically far. We have that
\begin{align*}
            \Delta( \rho_{A, \psi}^{k(\lambda)},\sigma_{A, \phi}^{k'(\lambda)}) \geq \frac{\abs{S(\rho_{A, \psi}^{k(\lambda)})- S(\sigma_{A, \phi}^{k'(\lambda)})} - c}{2 n (\lambda)} = \frac{\abs{E\big(\ket{\psi^{k(\lambda)}}_{AB}\big)-E\big(|\phi^{k'(\lambda)}\rangle_{AB}\big)} - c}{2 n (\lambda)} \, ,
\end{align*}
where the first inequality follows from \Cref{lem:Fannes}, and the last equality
  uses the definition of the entanglement entropy. Therefore, using the fact that for all $k,k' \in \{0,1\}^{\kappa(\lambda)}$ the entanglement entropies of pairwise sampled states from the ensembles is  $ \eta \geq  1/2e +1/\poly(n)$,
\begin{align*}
    \Delta( \rho_{A, \psi}^{k(\lambda)},\sigma_{A, \phi}^{k'(\lambda)}) \geq
          \Omega\left(\frac{1}{\poly (\lambda)}\right) \quad \forall k, k' \in \{0,1\}^{\kappa(\lambda)} \, .
\end{align*}

\end{proof}

Every known construction of pseudoentanglement from pure state ensembles
\cite{ABF+23,ABV23,LREJ25} exhibits an entanglement entropy gap of at least
$1/2e+1/\poly(n)$, making them suitable for constructing EPFI pairs. While the definitions in \cite{ABV23,LREJ25} allow for a larger gap when
considering computationally efficient entanglement measures, it is the information-theoretic measure of entanglement entropy that determines the relationship between the gap and trace distance, and thus enables the construction of EPFI pairs. The connection between information-theoretic and computationally meaningful entanglement measures, along with its cryptographic implications, is explored in more detail in the following section.

\begin{remark}
  We notice that in the proof of \Cref{lem:polypseudo}, we do not need
  indistinguishability between the pure states, but only of subsystem $A$ (or
  $B$). 
  In this case, we can also achieve EPFI pairs under a weaker notion of
  pure-state pseudoentanglement.
\end{remark}

The use of entanglement entropy as a measure of entanglement in the case of pure states require a lower bound of $\omega(\log n))$ for the low entangled family. However, when computational measures of entanglement are taken into account, the entanglement gap can be even larger for pure states, i.e., $\Omega(n)$ vs. $o(1)$ for other entanglement measures as proven in \cite{LREJ25}. While it is true that entanglement entropy losses its operational meaning when it comes to quantify the distillable entanglement (or entanglement cost) taking into account computational efficiency, it is relevant for the construction of EPFI pairs.

\subsection{Mixed state pseudoentanglement implies EPFI pairs}

Let us now study the pseudoentanglement in the case of mixed states. The main result of this subsection is the construction of EPFI pairs from pseudoentangled mixed states. Let us first introduce our proposed definition of pseudoentanglement for mixed states.

\begin{definition} [Mixed $\eta$-gap pseudoentanglement] Let $\lambda \in \NN_{+}$ and $\eta : \NN_{+} \rightarrow \NN_{+} $ be arbitrary. A pair of families of mixed  bipartite states $\{\psi_{AB}^{k(\lambda)}\}_{k(\lambda)}$ and $\{\phi_{AB}^{k'(\lambda)}\}_{k'(\lambda)}$ indexed by $k, k' \in \{0,1\}^{\kappa(\lambda)}$ is said to have \textit{mixed $\eta$-gap pseudoentanglement} if : 
\begin{enumerate}
    \item \textbf{Efficient generation}: Given $k(\lambda)$ (or $k'(\lambda)$),  there exists a QPT algorithm A that on input $(1^{\lambda}, k,  b)$ (or $(1^{\lambda}, k', b)$) outputs $\psi_{AB}^{k(\lambda)}$ (or $\phi_{AB}^{k'(\lambda)}$). 
    \item \textbf{Entanglement gap}: For all $k, k' \in \{0,1\}^{\kappa(\lambda)}$,
    \begin{align*}
        \left|E^{\infty}_R \left( \psi_{AB}^{k(\lambda)}\right) - E^{\infty}_R \left( \phi_{AB}^{k'(\lambda)}\right) \right| \geq \eta \, .
    \end{align*}
   \item  \textbf{Computational indistinguishability}: For any non-uniform QPT distinguisher $\DD $ with advice $\sigma_{\lambda}$ and any $m \in \poly(\lambda)$, there exists a negligible function $\nu (\lambda) > 0$ such that:
    \begin{equation*}
        \left| \probP_{\rho \leftarrow \{\psi_{AB}^{k(\lambda)}\}_{k(\lambda)} } [\DD(\sigma_{\lambda},\rho^{\otimes m}) = 1] - \probP_{\rho \leftarrow \{\phi_{AB}^{k'(\lambda)}\}_{k'(\lambda)} }  [\DD( \sigma_{\lambda}, \rho^{\otimes m}) = 1]  \right | \leq \nu (\lambda) \, .
    \end{equation*}
\end{enumerate}
\label{def:M-PSE}
\end{definition}

Our definition is based on the ones proposed by \cite{ABV23,GE24} but with two
modifications. Foremost, both families of states have to be efficiently
generated, contrary to previous definitions in which only the low entangled
family was the efficiently generated one. The other major difference is that we consider the (regularised) relative entropy of entanglement, instead of
the computational distillable entanglement or the computational entanglement cost, which was used in those results.

Unlike in \cite{ABV23,GE24}, our proposed definition of pseudoentanglement is defined in the asymptotic IID setting without taking into account computational entanglement measures. As proposed in \cite{ABV23}, their pseudoentanglement construction can be extended to the computational asymptotic IID setting by taking into account taking into account the measures,
    \begin{align*}
        &\hat{E}_C^{\infty} (\rho_{AB}) = \inf_{\epsilon \in (0,1]} \lim_{t \rightarrow \infty} \sup \frac{1}{t} \hat{E}^{\epsilon}_C (\rho_{AB}^{ \otimes t})\, , \\
         &\hat{E}^{\infty}_D (\rho_{AB}) = \inf_{\epsilon \in (0,1]} \lim_{t \rightarrow \infty} \inf \frac{1}{t} E^{\epsilon}_D(\rho_{AB}^{ \otimes t}) \, ,
    \end{align*}
where $\hat{E}^{\epsilon}_C (\rho)$ and $\hat{E}^{\epsilon}_D(\rho)$ are defined in \Cref{def:comptentcost} and \Cref{def:cmpdistent}. Let us now show that our definition is a relaxation in the condition of the entanglement gap with respect to the ones of \cite{ABV23,GE24} in the asymptotic IID setting.

Given a “highly entangled” family $\{\phi_{AB}^{k'(\lambda)}\}_{k'(\lambda)}$ such that $\hat{E}^{\infty}_D\left( \{k',\phi_{AB}^{k'}\}\right) \geq d(\lambda)$ and a “low entangled”  family $\{\psi_{AB}^{k(\lambda)}\}_{k(\lambda)}$ such that $\hat{E}^{\infty}_C\left( \{k,\psi_{AB}^{k}\}\right) \leq c(\lambda)$ with $c(\lambda) < d(\lambda)$. Therefore, given two states $\phi_{AB}^{k'(\lambda)}$ and $\psi_{AB}^{k(\lambda)}$,
\begin{align*}
    d(\lambda) - c(\lambda) \leq \hat{E}^{\infty}_D \left( \{k',\phi_{AB}^{k'}\}\right)- \hat{E}^{\infty}_C\left( \{k,\psi_{AB}^{k}\}\right) \leq E_R^{\infty} ( \phi_{AB}^{k'(\lambda)}) - E_R^{\infty} (\psi_{AB}^{k(\lambda)})
\end{align*}
for all $k$ and $k'$, since for any family $\hat{E}_D^{\infty} \leq E_D^{\infty}\leq E_R^{\infty}\leq E_C^{\infty} \leq \hat{E}_C^{\infty}$ under LOCC operations.

\begin{corollary}[Mixed pseudoentanglement implies EPFI pairs]\label{lem:mixpolypseudo} 
For $\eta \geq 2 +1/\poly(n)$ and assuming the existence of mixed  $\eta$-pseudoentanglement, then EPFI pairs exist.    
\end{corollary}

Please, note that the construction is similar to the one of \Cref{th:PREFI}, but taking into account that in this case the studied resource is the entanglement. In the case of entanglement, the regularised relative entropy of resource is equivalent to the regularised relative entropy of entanglement, which asymptotically continuity bound is given by \Cref{lem:mixFannes}.

Therefore, every state that has a large pseudoentangled gap $d(\lambda)-c(\lambda) \geq 2 +1/\poly(n)$ are eligible for building EPFI pairs. On the other hand, there potentially  exist pairs of families of mixed states which does not present a gap in the computational measures of entanglement, i.e., they are not pseudoentangled following the definition of \cite{ABV23,GE24}, but from which EPFI pairs can be constructed. 

\subsection{Beyond EPFI: computationally locked entanglement}

Previous definitions of pseudoentanglement focus on efficiently generated states with low entanglement that are computationally indistinguishable from highly entangled states. This is an analog of the “classical” concept of pseudorandomness: a simple object (such as a pseudorandom string or pseudoentangled state) that can replace a complex one (resp. random string or highly entangled state).

On the other hand, given how central entanglement is in quantum information, and the plethora of applications that require highly entangled states, we can flip this question and ask for efficiently generated states with high entanglement that are computationally indistinguishable from low entangled ones (that may be efficiently generated or not). We denote this as {\em computationally locked entanglement}, and it can be used to conceal entanglement that could be distilled only with the help of a secret key. We notice that pseudoentangled states where both low- and high-entangled families are efficiently generated (as in our definition of pseudoentanglement)
exhibit computationally locked entanglement.

\begin{definition} [Computationally locked entanglement]
Let $\lambda \in \NN$, $n : \NN \rightarrow \NN$ be a polynomially bounded function, $\epsilon : \NN \rightarrow [0,1]$ and $c,d : \NN \rightarrow \NN $ with $c < d$. A family of $2n(\lambda)$-qubit bipartite states $\{\psi_{AB}^{k(\lambda)}\}_{k(\lambda)}$ is said to have \textit{computationally locked entanglement} $(\epsilon, c, d )$ if there is a family of $2n(\lambda)$-qubit bipartite states $\{\phi_{AB}^{k'(\lambda)}\}_{k'(\lambda)}$, such that:
\begin{enumerate}
    \item  The computational entanglement cost of the family $\{\phi_{AB}^{k'(\lambda)}\}_{k'(\lambda)}$ is upper bounded as $\hat{E}^{\epsilon}_C(\{k, \phi_{AB}^{k'}\}) \leq c$.
    \item Given $k(\lambda)$,  there exists a QPT algorithm A that on input $(1^{\lambda}, k,  b)$  outputs $\psi_{AB}^{k(\lambda)}$.
    \item The computational distillable entanglement  of the family $ \{\psi_{AB}^{k(\lambda)}\}_{k(\lambda)}$ is lower bounded as $\hat{E}^{\epsilon}_D (\{k, \psi_{AB}^{k}\})\geq d$.
    \item  For any non-uniform QPT distinguisher $\DD $ with advice $\sigma_{\lambda}$ and any $m \in \poly(\lambda)$, there exists a negligible function $\nu (\lambda) > 0$ such that:
    \begin{equation*}
        \left| \probP_{\rho \leftarrow \{\psi_{AB}^{k(\lambda)}\}_{k(\lambda)} } [\DD(\sigma_{\lambda},\rho^{\otimes m}) = 1] - \probP_{\rho \leftarrow \{\phi_{AB}^{k'(\lambda)}\}_{k'(\lambda)} }  [\DD( \sigma_{\lambda}, \rho^{\otimes m}) = 1]  \right | \leq \nu (\lambda) \, .
    \end{equation*}
\end{enumerate}
\label{def:CLE}
\end{definition}

Computational locked entanglement can be viewed as a ``dual" notion to pseudoentanglement, as defined in \cite{ABV23,GE24}. An example of computational locked entanglement is the construction of pseudorandom density matrices \cite{BMB+24}. However, in this case, the "low-entangled" family is also efficiently generated. It would be interesting to investigate how pseudoentanglement and computationally locked states relate to each other, i.e., if one implies the other, or if they are incomparable.

We now prove the intuitive consequences of computationally locked states for parties that do not have access to the key.

\begin{lemma} Given a family of $2n(\lambda)$-qubit bipartite states $\{\psi_{AB}^{k(\lambda)}\}_{k(\lambda)}$ that have computationally locked entanglement $(\epsilon, c, d )$, where $c < d$, its computational distillable entanglement without access to the key $k$ is upper bounded as $\hat{E}^{\epsilon}_D (\{\psi_{AB}^{k}\}) \leq c$, as defined in \Cref{def:entanglement-distillable-wo-key} .
\end{lemma} 

\begin{proof} 

It can be proven by contradiction. Suppose there exist a family of states $\{\psi_{AB}^{k(\lambda)}\}_{k(\lambda)}$ that have computationally locked entanglement as defined in \Cref{def:CLE} such that, not given the keys $k$, $\hat{E}^{\epsilon}_D (\{ \psi_{AB}^{k}\}) > c$. As proven in \cite{ABV23}, $\hat{E}^{\epsilon}_D < \hat{E}^{\epsilon}_C$. Moreover, $\hat{E}^{\epsilon}_C (\{k', \phi_{AB}^{k'(\lambda)}\}) < c$ by construction. Then, if there is a poly time algorithm that is able to distill $\hat{E}^{\epsilon}_D (\{ \psi_{AB}^{k(\lambda)}\}) > c$, it would be able to distinguish between the pair of families  $\{\psi_{AB}^{k(\lambda)}\}_{k(\lambda)}$ and $\{\phi_{AB}^{k'(\lambda)}\}_{k'(\lambda)}$ which is not possible by definition. 
\end{proof}

\begin{corollary}
Given a $2n(\lambda)$-qubit bipartite states
  $\{\psi_{AB}^{k(\lambda)}\}_{k(\lambda)}$ which has \textit{computationally locked entanglement} $(\epsilon, d )$, i.e., w.r.t. a family of $2n(\lambda)$-qubit separable bipartite states $\{\phi_{AB}^{k'(\lambda)}\}_{k'(\lambda)}$, no entanglement can be distilled without the respective keys $k$.
  \end{corollary}

An application of the computationally locked functionality is an authenticated quantum teleportation protocol. The scheme is similar to the original quantum teleportation \cite{BBC+93} but in this case the receiver cannot access the teleported state without using a secret key in each interaction, being not necessary to authenticate the classical channel before or during the teleportation protocol. Moreover, the family of states with its corresponding keys can be used a polynomial number of times, unlike in the Clifford encryption scheme \cite{DLT02,ABOE08}.

Further applications of computationally locked entanglement for quantum networks in which the there is a necessity of distributing entanglement while preventing the users for accessing it. Since quantum networks routing is based on the principle of entanglement swapping \cite{BCZ98,Ca18}, encoding the nodes of the quantum network with computationally locked entanglement allows certifying the routing of the network.

\section*{Acknowledgments}

 Álvaro Yángüez thanks Pere Munar, Lorenzo Leone, Asad Raza, Ludovico Lami, Salvatore F.E. Oliviero and Francesco Anna Mele for helpful discussions. 
 Alex B. Grilo and Álvaro Yángüez are supported by the European Union's Horizon Europe Framework Programme under the Marie Sklodowska Curie Grant No. 101072637, Project Quantum-Safe Internet (QSI).

\bibliography{bibliography.bib}

\end{document}